\documentclass{article}

\usepackage{subfigure}
\usepackage{cmll}
\usepackage{MnSymbol}
\usepackage{stmaryrd}
\usepackage{bussproofs}
\usepackage{bbold}
\newcommand\+{\mathrm{+}}
\newcommand\rta{\rightarrow}
\newcommand\Rta{\Rightarrow}
\newcommand\Lra{\Leftrightarrow}
\newcommand\llb{\llbracket}
\newcommand\rrb{\rrbracket}

\newcommand\llc{\llparenthesis}
\newcommand\rrc{\rrparenthesis}

\newcommand\cons{\mathrm{::}}

\newcommand\Lcalcul{${\lambda}\!$-calculus }
\newcommand\dlcalcul{$\partial\!\lambda\!$-calculus }
\newcommand\tdlcalcul{$\tau\!\partial\!\lambda\!$-calculus }

\newcommand\Dinf{\mathcal{D}_\infty}
\newcommand\Minf{\mathcal{M}_\infty}

\newcommand\lolipop{\multimap}

\newtheorem{definition}{Definition}
\newtheorem{theorem}{Theorem}
\newtheorem{lemma}{Lemma}

\newtheorem{fact}{Fact}

\newenvironment{proof}{{\bf Proof. }}{$\ \ \ \ \ \ \square$\\ \smallskip}

\title{On the discriminating power of tests in resource $\lambda$-calculus}

\begin{document}
\maketitle

%\tableofcontents
\begin{abstract}
Since its discovery, differential linear logic (DLL) inspired numerous domains. In denotational semantics, categorical models of DLL are now commune, and the simplest one is Rel, the category of sets and relations. In proof theory this naturally gave birth to differential proof nets that are full and complete for DLL. In turn, these tools can naturally be translated to their intuitionistic counterpart. By taking the co-Kleisly category associated to the $!$ comonad, Rel becomes MRel, a model of the \Lcalcul that contains a notion of differentiation. Proof nets can be used naturally to extend the \Lcalcul into the lambda calculus with resources, a calculus that contains notions of linearity and differentiations. Of course MRel is a model of the \Lcalcul with resources, and it has been proved adequate, but is it fully abstract?

That was a strong conjecture of Bucciarelli, Carraro, Ehrhard and Manzonetto in \cite{BCEM11}. However, in this paper we exhibit a counter-example. Moreover, to give more intuition on the essence of the counter-example and to look for more generality, we will use an extension of the resource \Lcalcul also introduced by Bucciarelli {\em et al} in \cite{BCEM11} for which $\Minf$ is fully abstract, the tests.

\end{abstract}
%\begin{keyword}
%Full abstraction, resource $\lambda$-calculus, relational semantics, Scott's semantics, non deterministic $\lambda$-calculus.
%\end{keyword}

%\newpage

\section{Introduction}
The first extension of the \Lcalcul with resources, by Boudol in 1993 \cite{Boudol93}, was introducing a special resource sensitive application that may involve multisets of affine arguments (each one has to be used at most one time). This was a natural way to export resource sensitiveness to the functional setting. However, gathering no known and interesting properties (confluence, linearity...), it was not fully explored. 

Later on, Ehrhard and Regnier, working on functional interpretation of differential proof nets, discovered a calculus similar to Boudol's one, named differential \Lcalcul \cite{EhRe04}. By adding to the $\lambda$-calculus a derivative operation $\frac{\partial M}{x}(N)$, which syntactically corresponds to a linear substitution of $x$ by $N$ in $M$, it recovers the resource-sensitiveness. This is done through the translation $(\lambda x.M)\ [N_1,...,N_n;N^!]\simeq(\lambda x.\frac{\partial^{\!(\!n\!)\!}M}{x}\!(\!N_1\!,\!...\!,N_n\!))\ N$ where $N_i$ are the linear arguments and $N$ is the non linear one. This more semantical view even allow for the generalisation of the operation and recover excellent semantical properties (confluence, Taylor expansion...). We will adopt the syntax of \cite{PaTr09} that re-implements improvements from differential $\lambda$-calculus into Boudol's calculus, and we will call it $\partial\lambda$-calculus.

The category Rel of set and relations is known to model the linear logic, and, despite its high degree of degeneration ($Rel^{op}=Rel$), it is a very natural construction. Indeed, what appeared to be a degeneration is in reality a natural choice that preserves all proofs, {\em i.e.} the interpretation function from proof to MRel is injective up to isomorphism (\cite{CaTo10}). But our principal interest for this category is that it models the differential linear logic, and of known such category it is the simplest and more natural.

As for every categorical model of linear logic, the interpretation of the $!$ induced a comonad. From that comonad we can construct the co-kleisly category. In the case of Rel, this new category, MRel, corresponds to the category of sets with, as morphisms from $A$ to $B$, the relations from $\mathcal{M}_f(A)$ (the finite multisets over $A$) to $B$. It is then a model of the $\lambda$-calculus and of the $\partial\lambda\!$-calculus. This construction being the most natural we can do, MRel is, {\em a priori}, one of the most natural models of the $\partial\lambda$-calculus (even if non well pointed \cite{BEM07}). 

It is only natural, then, to question on the depth of the link among the reflexive elements of MRel and the $\partial\lambda\!$-calculus. And more precisely among MRel's canonical reflexive element $\mathcal{M}_\infty$ and the $\partial\lambda\!$-calculus. Until now we knew that $\Minf$ (\cite{BEM07}) was adequate for the $\lambda$-calculus,  {\em i.e.} that two terms carrying the same interpretations in MRel behave the same way in all contexts. But we did not know anything about the counterpart, named full abstraction. 

This question has been thoroughly studied, however, since $\Minf$ has been proved (resp. in \cite{Man09}, \cite{BCEM12} and \cite{BCEM11}) fully abstract not only for both of the principal sub-calculi of $\partial\lambda$-calculus, namely the usual $\lambda$-calculus and Kfoury's linear calculus of \cite{Kfo00}, but also for the extension with tests of \cite{BCEM11} (denoted $\tau\partial\lambda$-calculus). Therefore Bucciarelli {\em et al} emit in \cite{BCEM11} a strong conjecture of full abstraction for the $\partial\lambda$-calculus.

However, and it is our purpose here, a counter example can be found. In order to exhibit this counter example, we will take an unusual shortcut using full abstraction result for $\tau\partial\lambda$-calculus. Indeed, we will prove a slightly more general theorem: the failure of full abstraction for $\partial\lambda$-calculus of any model that is fully abstract for $\tau\partial\lambda$-calculus. Due to this generalization we will not have to introduce the full description of $\Minf$ in the core of the article (it is available in annexes).

Additionally to be considerably easier and more intuitive than the direct and usual method, this way of proceeding is part of a larger study of full abstraction. Indeed, we are looking for a mechanical way to tackle full abstraction problems in two steps. First we extend the calculus with well chosen semantical objects in order to reach the definability of compact elements. Then we study the full abstraction question indirectly via the link between the operational equivalence of the original calculus and of its artificial extension. This reduces our mix of semantic and syntactic question to a purely syntactic one, allowing us to use powerful syntactic constructions.

Tests where introduced in \cite{BCL99} to have a full abstraction theorem for Boudol $\lambda$-calculus with resources. Later on,
the principle was improved in \cite{BCEM11}, implementing semantic objects in the syntax in order to get full abstraction of $\Minf$, following an idea of \cite{EL10}. This extension can, in our context, be compared to a basic exception mechanism. The term $\bar\tau(Q)$ is raising the exception (or test) $Q$, absorbing all its non resources-limited applications and the exception $\tau(M)$ is catching any exception in $M$ by annihilating all the head-$\lambda$-abstraction. The most important here being the scope of the $\tau(M)$ that act as an infinite application over $M$.

\vspace{-15pt}

\paragraph*{Notations:}
\noindent $\lambda \bar x^n$ will be used for $\lambda x_1,...,x_n$ ($n$ is not specified when it can be any integer) % and $M\ \bar N^n$ for $M\ N\ \cdots\ N$ with $n$ copies of $N$ 
and $I$ will denotes the identity $\lambda x.x$.

\section{Background}
\label{1}

\subsection{\dlcalcul}
As explained, this article is directly following \cite{BCEM11}. For this reason we need to introduce the \dlcalcul and then the tests. In the \dlcalcul, the notion of linearity is capital. Any term in linear position will never suffer any duplication or erasing regardless the reduction strategy. Linear subterms are subterms that are either the first subterm of a lambda abstraction in linear position, the left side of an application that is in linear position, or in the linear part of its right side. The last case is the real improvement and asks for arguments to be separated in linear and non linear arguments. Therefore, the right side of the applications will be replaced by a new kind of expression different from terms, the ``bags''. Bags are multisets containing some linear (non banged) arguments and exactly one non linear (banged) argument:\vspace{-8pt}
\begin{center}
\begin{tabular}{l c l}
(terms) & $M,N :$ & $\lambda x.M\ |\ M\ B\ |\ M\+ N\ |\ 0$\\
(bags) & $B,C :$ & $[M_1,...,M_n;M^!]\ |\ B+C\ |\ 0$\\
\end{tabular}
\end{center}
\noindent This is the syntax of \cite{BCEM11} modulo the macro $M\+N=(\lambda x.x)\ [\{M\+N\}^!]$.\\
For convenience, the finite sums will be denoted $\Sigma_iQ_i$ and the different $0$'s are just the neutral elements of the different sums. This demonic sum had to be implemented since we want the calculus to be resource sensitive and confluent, thus there is no other choice than to considere the sum of all the possible outcomes. Sums distribute with any linear context:
\begin{center}
$\lambda x.(\Sigma_i M_i) = \Sigma_i(\lambda_x.M_i)$
\hskip 70pt
$(\Sigma_i M_i)\ (\Sigma_j N_j) = \Sigma_{i,j} (M_i\ N_j)$
\hskip 70pt
$([(\Sigma_{i_1\le k_1} M_{i_1}^1),...,(\Sigma_{i_n\le k_n} M_{i_n}^n);M^!])\ =\ 
 \Sigma_{(i_j)_j\le (k_j)_j}[ M_{i_1}^i,...,M_{i_n}^n;M^! ]$
\end{center}
In the application, each linear argument will replace one and only one occurrence of the variable, thus the need of two kinds of substitutions, the usual one, denoted $\{.\}$, and the linear one, denoted $\langle.\rangle$. This last will act like a derivation $M\langle N/x\rangle\sim \frac{\partial M}{\partial x}(N)$:
\begin{center}
$x\langle N/x\rangle = N$
\hskip 50pt
$x\langle N/y\rangle = 0$
\hskip 50pt
$(\lambda y.M)\langle N/x\rangle = \lambda y.(M\langle N/x\rangle)$

$([M_1, ...,M_n; M^!]) \langle N/x\rangle~=~(\Sigma_{i=1}^n [M_1,..,M_i\langle N/x\rangle,
 ...M_n;M^!])~+~[M_1,...,M_n,M\langle N/x\rangle;M^!]$

$(M\ P)\langle N/x\rangle=(M\langle N/x\rangle\ P)+(M\ P\langle N/x\rangle)$
\end{center}

\noindent This enables us to describe the $\beta$-reduction:\vspace{-5pt}
\begin{eqnarray*}
(\beta)\ \ &\ \ (\lambda x.M) [N_1,...,N_n;N^!] &\rta M\langle N_1/x\rangle\cdots\langle N_n/x\rangle\{N^!/x\}\vspace{-8pt}
\end{eqnarray*}
In other words $(\lambda x.M) [N_1,...,N_n;N^!]\rta \frac{\partial^nM}{\partial^n x}(N_1,...,N_n)(N)$

\subsection{\tdlcalcul}
In differential proof nets the 0-ary tensor and the 0-ary par can be added freely in the sense that we still have a natural interpretation in MRel and $\Minf$. These operations can be translated in our calculus as an exception mechanism. With on one side a $\tau(Q)$ that ``raises'' the exception (or test) $Q$ by burning its applicative context (whenever these applications do not have any linear component, otherwise it diverges). And with on the other side a $\bar\tau(M)$ that ``catch'' the exceptions in $M$ by burning the abstraction context of $M$ (whenever this abstraction is dummy). The main difference with a usual exception system is the divergence of the catch if no exception are raised.\\
We introduce a new operators and a new kind of expression that will play the role of exception, the tests:\vspace{-8pt}
\begin{eqnarray*}
\textrm{(terms)} & M,N : \bar\tau(Q)\\
\textrm{(test)} & Q,R : & \ \epsilon\ |\ Q|R\ |\ \tau (M)\ |\ Q\+R\ |\ 0
\end{eqnarray*}
New operator immediately imply new distribution rules for the sum and the linear substitution:
\begin{center}
$\tau(\Sigma_i M_i) = \Sigma_i\tau(M_i)$
\hskip 30pt
$\bar\tau(\Sigma_{i} Q_i) =  \Sigma_i\bar\tau(Q_i)$
\hskip 30pt
$\|_j \Sigma_{i} Q_{i(j)} = \Sigma_i \|_j Q_{i(j)}$\\
$\tau(M)\langle N/x\rangle = \tau(M\langle N/x\rangle)$
\hskip 30pt
$\bar\tau(Q)\langle N/x\rangle = \bar\tau(Q\langle N/x\rangle)$\\
%\hskip 30pt
$(Q\+R)\langle N/x\rangle = Q\langle N/x\rangle\+R\langle N/x\rangle$
\hskip 30pt
$(Q|R)\langle N/x\rangle = Q\langle N/x\rangle|R\langle N/x\rangle$
\end{center}

\noindent Here is the corresponding operational semantics:
\begin{eqnarray*}
(\gamma)\ \ &\ \ \tau[\bar\tau(Q)]\ &\rta\ Q\\
(\tau)\ \ &\ \ \tau (\lambda x.M)\ &\rta\ \tau (M\{0/x\})\\
(\bar\tau_1)\ \ &\ \ (\bar\tau(Q))\ [M^!]\ &\rta\ \bar\tau (Q)\\
(\bar\tau_2)\ \ &\ \ (\bar\tau(Q))\ [M_1,...,M_{n\ge 1};M^!]\ &\rta\ 0\\
(\epsilon)\ \ &\ \ \epsilon | \epsilon\ &\rta\ \epsilon
\end{eqnarray*}
%\begin{center}
%$(\gamma)\ \ \tau[\bar\tau(Q)]\ \rta\ Q$
%\hskip 50pt
%$(\tau)\ \ \tau (\lambda x.M)\ \rta\ \tau (M\{0/x\})$
%\hskip 50pt
%$(\bar\tau_1)\ \ (\bar\tau(Q))\ [M^!]\ \rta\ \bar\tau (Q)$
%\hskip 50pt
%$(\bar\tau_2)\ \ (\bar\tau(Q))\ [M_1,...,M_{n\ge 1};M^!]\ \rta\ 0$
%\hskip 50pt
%$(\epsilon)\ \ \epsilon | \epsilon\ \rta\ \epsilon$
%\hskip 50pt
%$(M+0)\ \ M\+0\ \rta\ M$
%\hskip 50pt
%$(Q+0)\ \ Q\+0\ \rta\ Q$
%\hskip 50pt
%$(B+0)\ \ B\+0\ \rta\ B$
%\hskip 75pt
%$(@+_2)\ \ (M)\ (B\+C)\ \rta\ (M\ B)\+(M\ C)$
%\\
%$(\lambda+)\ \ \lambda x.M\+N\ \rta\ (\lambda x.M)\+(\lambda x.N)$
%\hskip 45pt
%$(@+_1)\ \ (M\+N)\ B\ \rta\ (M\ B)\+(N\ B)$
%\hskip 50pt
%$(lin+)\ \ [\Sigma_{i_1}M_{i_1}^1,...,\Sigma_{i_n}M_{i_n}^n;M^!]\ \rta\ \Sigma_{i_1....i_n}[M_{i_1},...,M_{i_n};M^!]$
%\\
%$(\tau+)\ \ \tau(M+N)\ \rta\ \tau(M)+\tau(N)$
%\hskip 50pt
%$(\bar\tau+)\ \ \bar\tau(Q+R)\ \rta\ \bar\tau(Q)+\bar\tau(R)$
%\hskip 50pt
%$(|+)\ \ (\Sigma_i Q_i)|(\Sigma_j R_j)\ \rta\ \Sigma_{i,j}(Q_i|R_j)$
%\end{center}
The intuition of $\bar\tau(Q)$ is an operator that take a test (a Boolean value), compute it and returns an infinite $\lambda$-abstraction with no occurrence of the abstracted variables. The test $\tau(M)$ is taking a term and returns a successful test if the term is converging in a context that consists of an infinite empty application.

\subsection{Observational order and full abstraction}

In order to ask for full abstraction, one has to specify a reduction strategy. A natural choice would be the head reduction, but this would make $(\lambda x.x\ [0,0])$ a normal form while no applicative instantiation of $x$ allow the convergence of this term. Therefore, the reduction strategy we are considering will not be head-reduction, but the outer-head-reduction. This reduction will reduce subterms in linear position after the subterms in head positions (\cite{PaRo10}). The corresponding normal forms are terms and tests of the form:\vspace{-5pt}
\begin{center} $M+ \lambda\bar x. y\ [N_{1,1},...,N_{1,k_1};L_1^{!}]\ \cdots\ [N_{n,1},...,N_{n,k_n};L_n^{!}]$\\
$M+ \lambda\bar x. \bar\tau(Q)$\\
$Q+ (\| \tau(N_i))$
\end{center}\vspace{-3pt}
\noindent Where every $N_{\_}\!$ and $Q$ must be in outer-head normal forms and can't be a sum (but the $L_i\!$ are of any kind).

\begin{definition}
$M$ is observationally below $N$, if for all context $C\llc.\rrc$, we have $C\llc N\rrc$ which is outer-head-converging whenever $C\llc M\rrc$ is outer-head-converging.\\
They are observationally equivalent if moreover $N$ is observationally below $M$
\end{definition}
\noindent In the particular case of the \tdlcalcul, we can easily restrict contexts to test-contexts, which is contexts whose output is a tests. This will be applied systematically for simplification.\\
We will denote $\le_{\tau\partial}$ and $\equiv_{\tau\partial}$ the observational order and equivalence of the \tdlcalcul and $\le_{\partial}$ and $\equiv_{\partial}$ those of the \dlcalcul.

Bucciarelli, Carraro, Ehrhard and Manzonetto were then able to prove a strong theorem relating the model to the calculus:

\begin{theorem}
$\Minf$ is fully abstract for the $\Lambda$-calculus with resources and tests: for all closed terms $M,N$ with resources ans tests,
 $$\llb M\rrb=\llb N\rrb\ \Lra\ M\equiv_{\tau\partial} N\vspace{-5pt}$$
\end{theorem}

%The implication is trivial with the Taylor expansion and the reverse implication was proven by definissability of prime elements.

\section{The counter-example}
\label{4}
In order to exhibit our counter-example we will use the following property:

\begin{fact}
Let $\mathcal{B}$ a calculus and $\mathcal{A}$ a super-calculus. Let $\mathcal{M}$ a model that is fully abstract for $\mathcal{A}$. $\mathcal{M}$ is fully abstract for $\mathcal{B}$ iff the operational equivalences of $\mathcal{B}$ and $\mathcal{A}$ are equal on their domain intersection.
\end{fact}

\noindent In our context it means that, in order to prove the non full abstraction for the \dlcalcul, it is sufficient to find two terms of the \dlcalcul that cannot be separated by any context of the \dlcalcul but that are separated by a context of the \tdlcalcul. This makes the research and the proof quite easier when the terms  of the \dlcalcul involved are complex but not the context of the \tdlcalcul.

\medskip

\noindent We are firstly exhibiting a term $A$ of the \dlcalcul that is observationally above the identity in the \dlcalcul', but not in the \tdlcalcul's observational order:\vspace{-5pt}
\begin{equation}
A=\Theta\ [\lambda uvw.w\ [I\ [v^!]]\ ,\ (\lambda uvw.u\ [(v\ [w^!])^!])^!] \label{eq1} \vspace{-5pt}
\end{equation}
where $\Theta$ is the Turing fix point combinator:\vspace{-5pt}
 $$(\lambda gu. u [(g\ [(g\ [u^!])^!])^!])\ [(\lambda gu. u [(g\ [(g\ [u^!])^!])^!])^!]\vspace{-5pt}$$
This term seems quite complex, but, modulo $\eta$-equivalence, $A$ reduces exactly to any of the elements of the following sum, and thus can be think as an equivalent:\vspace{-5pt}
\begin{equation}
\Sigma_{n\ge 1}B_n=\Sigma_{n\ge 1}\lambda \bar u^nw. w\ [I\ [u_1^!]\ [u_2^!]\ \cdots\ [u_n^!]] \label{eq2}\vspace{-5pt}
\end{equation}
This is due to the following property:

\begin{lemma}
\label{lemme5.3}
If $A_i=\lambda\bar x^{i+1}.A\ (x_1\ [x_2^!]\ \cdots\ [x_{i+1}^!])$ then $A_0\rta_\eta A$ and for all $i$, \vspace{-5pt}
$$A_{i}\rta A_{i+1}+B_{i+1}\vspace{-5pt}$$
\end{lemma}
\begin{proof}
Simple reduction unfolding the $\Theta$ once.
\end{proof}

\noindent In absence of tests, this term has a comportment similar to $\epsilon_0$ in the sense that it will converges in any applicative context provided that these applications do not carry linear components. In particular it converges more often than the identity:\begin{lemma}
For all context $C\llc.\rrc$ of the \dlcalcul, if $C\llc I\rrc$ converges then $C\llc A\rrc$ converges, {\em i.e.} $I\le_\partial A$
\end{lemma}
\begin{proof}
Let $C\llc.\rrc$ a context that converge on $I$\\
With the context lemma (\cite{PaRo10}), and since neither $I$ nor $A$ has free variables, we can assume that $C\llc.\rrc=\llc.\rrc\ P_1\ \cdots\ P_k$, thus by lemma \ref{lemme5.3}, we have $A\rta^* U+B_k$ and \vspace{-5pt}
\[C\llc A\rrc\rta^* U' + \lambda w.w\ [I\ P_1\ \cdots\ P_k]=U'+\lambda w.w\ [C\llc M\rrc]\vspace{-8pt}\]
converges.
\end{proof}

\noindent But in presence of real tests, its comportment appeared to be different that $\epsilon_0$ in the sense that it diverges under a $\bar\tau$. In particular it is not observationally above the identity in \tdlcalcul:

\begin{lemma}
\label{lemme 2}
In the \tdlcalcul, $\tau(A\ [\epsilon_0^!])$ diverges, while $\tau(I\ [\epsilon_0^!])$ converges, {\em i.e.} $I\not\le_{\tau\partial} A\vspace{-5pt}$
\end{lemma}
\begin{proof}
For all $i$, $\tau (A_i\ [\epsilon_0^!])$ diverges since, by co-induction:\vspace{-5pt}
\begin{eqnarray*}
\tau(A_i\ [\epsilon_0]) &\rta& \tau ((A_{i+1}+(\lambda\bar x^{i+1}y.y\ [I\ [x_1^!]\ \cdots\ [x_{i+1}^!]]))\ [\epsilon_0])\\
                       &=& \tau (A_{i+1}\ [\epsilon_0]) + \tau((\lambda\bar x^{i+1}y.y\ [I\ [x_1^!]\ \cdots\ [x_{i+1}^!]])\ [\epsilon_0])\\
                       &\rta& \tau (A_{i+1}\ [\epsilon_0]) +  \tau (\lambda\bar x^{i}y.y\ [\epsilon_0\ [x_1^!]\ \cdots\ [x_{i+1}^!]])\\
                       &\rta^*& \tau (A_{i+1}\ [\epsilon_0]) +  \tau (0\ [\epsilon_0\ [0^!]\ \cdots\ [0^!]])
\end{eqnarray*}
The non outer-head convergence comes with the co-induction hypothesis for the first term, and is trivial for the second.
\end{proof}

\vspace{-15pt}

\noindent Hence, we have broken the conjecture concerning the equality between the observational and denotational orders. Let's break the whole conjecture:

\begin{theorem}
$\Minf$ is not fully abstract for the $\lambda$-calculus with resources
\end{theorem}

\section{Further works}
First a diligent reader will remark that we have a critical use of the demonic sum which is very powerful in this calculus. And an even more diligent one will remark that an arbitrary choice have been made concerning this sum: we could differentiate terms and reduced of terms and remove sums from the original syntax (they just have to appear in reductions of terms). The choice we made here corresponds to the one of \cite{BCEM11} and carries an understandable counter-example. But we claim that another equivalent counter-example arises for the case with limited sum. This counter-example is, however, a little more complicated and make it necessary to rework the material of \cite{BCEM11} (even if everything works exactly the same way).

Our counter-example can be translated to some related cases. In particular, to prove non full abstraction of Scott's $\Dinf$ for the \Lcalcul with angelic and demonic sums (conjectured in \cite{BEM09}). For this calculus the extension with tests exists and is fully abstract for $\Dinf$, this is a trivial modification of the tests of \cite{Bre12a} (using general demonic and angelic sums). In this framework the term $\theta(\lambda xy.x+y)$ plays exactly the role of $A$ in our example with the same output.

In the end, from a unique object that is DLL, we exhibit two natural constructions, one in the semantical world, the other in the syntactical one, but they appeared do not respect full abstraction. One would say that they are not that natural and that more natural one may be found. But this would be to easy, from the state of art we don't known more natural construction. The misunderstanding comes with the concept of naturality, it seems that the syntactic idea of ``convergence'' does not really correspond to the equivalent in semantical word. One being a lowest fix point and the second a largest one. This difference appears when working with the demonic sum that allow to check the convergence in unbounded applicative context.

Finally we presented tests as a general tool whose importance is above the role we gave them here. This result is interesting and important as it presents tests as useful tools to verify that full abstraction fails. But it remains a negative result that does not justify alone any real interest for them. Further works will then focus on presenting positive proofs of full abstractions that are using tests. Following this way we already submitted a revisited proof of full abstraction of the Scott's $\Dinf$ for the usual \Lcalcul \cite{Bre12a}.

%This counter-example introduce a new question: what could be a continuous model that would be fully abstract for the $\lambda$-calculus with sum or with resources ? A study of the counter-examples and of the failure of tests elimination will surely help in order to find these domains.% Especially one can be convinced, by studies on Krivines machines, that a fully abstract domain for the $\lambda$-calculus with sums would have a lifted bottom but no lifted top. 

%On the other side, it appears as interesting to find which theory $\Minf$ modelise in \dlcalcul. These are research in progress that seems related to a probabilistic refinement of sums and resources (\cite{EPT11}).% Indeed probabilities are highly refining the $\Lambda$-theories that add a way to discriminate infinite sums from the others.

\bibliographystyle{eptcs}
\bibliography{article}

\begin{thebibliography}{10}
\providecommand{\bibitemdeclare}[2]{}
\providecommand{\surnamestart}{}
\providecommand{\surnameend}{}
\providecommand{\urlprefix}{Available at }
\providecommand{\url}[1]{\texttt{#1}}
\providecommand{\href}[2]{\texttt{#2}}
\providecommand{\urlalt}[2]{\href{#1}{#2}}
\providecommand{\doi}[1]{doi:\urlalt{http://dx.doi.org/#1}{#1}}
\providecommand{\bibinfo}[2]{#2}

\bibitemdeclare{article}{Boudol93}
\bibitem{Boudol93}
\bibinfo{author}{Gerard \surnamestart Boudol\surnameend}
  (\bibinfo{year}{1993}): \emph{\bibinfo{title}{The lambda-calculus with
  multiplicities}}.
\newblock {\sl \bibinfo{journal}{INRIA Research Report 2025}}.

\bibitemdeclare{article}{BCL99}
\bibitem{BCL99}
\bibinfo{author}{P.~\surnamestart Boudol\surnameend}, \bibinfo{author}{P.-L.
  \surnamestart Curien\surnameend} \& \bibinfo{author}{C.~\surnamestart
  Lavatelli\surnameend} (\bibinfo{year}{1999}): \emph{\bibinfo{title}{A
  semantics for lambda calculi with resources}}.
\newblock {\sl \bibinfo{journal}{Mathematical Structures in Comput. Sci.
  (MSCS)}} \bibinfo{volume}{Vol. 9}, pp. \bibinfo{pages}{437--482}.

\bibitemdeclare{unpublished}{Bre12a}
\bibitem{Bre12a}
\bibinfo{author}{Flavien \surnamestart Breuvart\surnameend}
  (\bibinfo{year}{2012}): \emph{\bibinfo{title}{A new proof of the
  Hyland/Wadsworth full abstraction theorem}}.
\newblock \bibinfo{note}{Submited}.

\bibitemdeclare{inproceedings}{BCEM11}
\bibitem{BCEM11}
\bibinfo{author}{Antonio \surnamestart Bucciarelli\surnameend},
  \bibinfo{author}{Alberto \surnamestart Carraro\surnameend},
  \bibinfo{author}{Thomas \surnamestart Ehrhard\surnameend} \&
  \bibinfo{author}{Giulio \surnamestart Manzonetto\surnameend}
  (\bibinfo{year}{2011}): \emph{\bibinfo{title}{Full Abstraction for Resource
  Calculus with Tests}}.
\newblock In \bibinfo{editor}{Marc \surnamestart Bezem\surnameend}, editor:
  {\sl \bibinfo{booktitle}{Computer Science Logic (CSL'11) - 25th International
  Workshop/20th Annual Conference of the EACSL}}, {\sl \bibinfo{series}{Leibniz
  International Proceedings in Informatics (LIPIcs)}}~\bibinfo{volume}{12},
  \bibinfo{publisher}{Schloss Dagstuhl--Leibniz-Zentrum fuer Informatik},
  \bibinfo{address}{Dagstuhl, Germany}, pp. \bibinfo{pages}{97--111}.

\bibitemdeclare{unpublished}{BCEM12}
\bibitem{BCEM12}
\bibinfo{author}{Antonio \surnamestart Bucciarelli\surnameend},
  \bibinfo{author}{Alberto \surnamestart Carraro\surnameend},
  \bibinfo{author}{Thomas \surnamestart Ehrhard\surnameend} \&
  \bibinfo{author}{Giulio \surnamestart Manzonetto\surnameend}
  (\bibinfo{year}{2012}): \emph{\bibinfo{title}{Full Abstraction for Resource
  Lambda Calculus with Tests, throught Taylor expansion}}.
\newblock \bibinfo{note}{Accepted}.

\bibitemdeclare{inproceedings}{BEM07}
\bibitem{BEM07}
\bibinfo{author}{Antonio \surnamestart Bucciarelli\surnameend},
  \bibinfo{author}{Thomas \surnamestart Ehrhard\surnameend} \&
  \bibinfo{author}{Giulio \surnamestart Manzonetto\surnameend}
  (\bibinfo{year}{2007}): \emph{\bibinfo{title}{Not Enough Points Is Enough}}.
\newblock In \bibinfo{editor}{Jacques \surnamestart Duparc\surnameend} \&
  \bibinfo{editor}{Thomas~A. \surnamestart Henzinger\surnameend}, editors: {\sl
  \bibinfo{booktitle}{CSL'07: Proceedings of 16\textsuperscript{th} Computer
  Science Logic}}, {\sl \bibinfo{series}{Lecture Notes in Computer Science}}
  \bibinfo{volume}{4646}, \bibinfo{publisher}{Springer}, pp.
  \bibinfo{pages}{268--282}.

\bibitemdeclare{inproceedings}{BEM09}
\bibitem{BEM09}
\bibinfo{author}{Antonio \surnamestart Bucciarelli\surnameend},
  \bibinfo{author}{Thomas \surnamestart Ehrhard\surnameend} \&
  \bibinfo{author}{Giulio \surnamestart Manzonetto\surnameend}
  (\bibinfo{year}{2009}): \emph{\bibinfo{title}{A relational model of a
  parallel and non-deterministic lambda-calculus}}.
\newblock In \bibinfo{editor}{Sergei~N. \surnamestart Art{\"e}mov\surnameend}
  \& \bibinfo{editor}{Anil \surnamestart Nerode\surnameend}, editors: {\sl
  \bibinfo{booktitle}{Logical Foundations of Computer Science, International
  Symposium, LFCS 2009}}, {\sl \bibinfo{series}{Lecture Notes in Computer
  Science}} \bibinfo{volume}{5407}, pp. \bibinfo{pages}{107--121}.

\bibitemdeclare{article}{CaTo10}
\bibitem{CaTo10}
\bibinfo{author}{Daniel \surnamestart de~Carvalho\surnameend} \&
  \bibinfo{author}{Lorenzo~Tortora \surnamestart de~Falco\surnameend}
  (\bibinfo{year}{2010}): \emph{\bibinfo{title}{The relational model is
  injective for Multiplicative Exponential Linear Logic (without weakenings)}}.
\newblock {\sl \bibinfo{journal}{CoRR}} \bibinfo{volume}{abs/1002.3131}.

\bibitemdeclare{article}{EL10}
\bibitem{EL10}
\bibinfo{author}{T.~\surnamestart Ehrhard\surnameend} \&
  \bibinfo{author}{O.~\surnamestart Laurent\surnameend} (\bibinfo{year}{2010}):
  \emph{\bibinfo{title}{Interpreting a finitary pi-calculus in differential
  interaction nets}}.
\newblock {\sl \bibinfo{journal}{Inf. Comput.}} \bibinfo{volume}{208(6)}, pp.
  \bibinfo{pages}{606--633}.

\bibitemdeclare{article}{EhRe04}
\bibitem{EhRe04}
\bibinfo{author}{Thomas \surnamestart Ehrhard\surnameend} \&
  \bibinfo{author}{Laurent \surnamestart Regnier\surnameend}
  (\bibinfo{year}{2004}): \emph{\bibinfo{title}{The differential
  lambda-calculus}}.
\newblock {\sl \bibinfo{journal}{Theoretical Computer Science, Elsevier}}.

\bibitemdeclare{article}{Kfo00}
\bibitem{Kfo00}
\bibinfo{author}{A.~J. \surnamestart Kfoury\surnameend} (\bibinfo{year}{2000}):
  \emph{\bibinfo{title}{A linearization of the Lambda-calculus and
  consequences}}.
\newblock {\sl \bibinfo{journal}{J. Log. Comput.}}
  \bibinfo{volume}{10}(\bibinfo{number}{3}), pp. \bibinfo{pages}{411--436}.
\newblock \urlprefix\url{http://dx.doi.org/10.1093/logcom/10.3.411}.

\bibitemdeclare{inproceedings}{Man09}
\bibitem{Man09}
\bibinfo{author}{Giulio \surnamestart Manzonetto\surnameend}
  (\bibinfo{year}{2009}): \emph{\bibinfo{title}{A general class of models of
  $\mathcal{H}^{\star}$}}.
\newblock In: {\sl \bibinfo{booktitle}{Mathematical Foundations of Computer
  Science (MFCS'09)}}, {\sl \bibinfo{series}{Lecture Notes in Computer
  Science}} \bibinfo{volume}{5734}, \bibinfo{publisher}{Springer}, pp.
  \bibinfo{pages}{574--586}.

\bibitemdeclare{article}{PaTr09}
\bibitem{PaTr09}
\bibinfo{author}{M.~\surnamestart Pagani\surnameend} \&
  \bibinfo{author}{Tranquilli \surnamestart P.\surnameend}
  (\bibinfo{year}{2009}): \emph{\bibinfo{title}{Parallel reduction in resource
  lambda-calculus}}.
\newblock {\sl \bibinfo{journal}{APLAS'09}} \bibinfo{volume}{5904 LNCS}, pp.
  \bibinfo{pages}{226--242}.

\bibitemdeclare{article}{PaRo10}
\bibitem{PaRo10}
\bibinfo{author}{Michele \surnamestart Pagani\surnameend} \&
  \bibinfo{author}{Simona Ronchi~Della \surnamestart Rocca\surnameend}
  (\bibinfo{year}{2010}): \emph{\bibinfo{title}{Linearity, Non-determinism and
  Solvability}}.
\newblock {\sl \bibinfo{journal}{Fundamenta Informaticae}}
  \bibinfo{volume}{103}(\bibinfo{number}{1-4}), pp. \bibinfo{pages}{173--202}.
\newblock \urlprefix\url{http://dx.doi.org/10.3233/FI-2010-324}.

\end{thebibliography}

\newpage

\appendix
\section{The model $\Minf$}
\subsection{Categorical model}
The category Rel of sets and relations is known to be a model of linear logic as it is a Seely category (we are giving the interpretation, but we will let the comutative diagrams to the reader since their comutations are trivial or like):

It is monoidal with tensor functor given by $A\otimes B=A\times B$, $f\otimes g=\{((u,x),(v,y))|(u,v)\in f, (x,y)\in g\}$ and with the arbitrary unit $1=\{*\}$. It is symetric monoidal close with $(A\lolipop B)=A\times B$ if we take the evaluation $ev=\{(((a,b),a),b)|a\in A,b\in B\}\in Rel((A\lolipop B)\otimes A,B)$. So it is star autonomus with $1$ as dualising object (for a trvial duality).\\
This give us the interpretation of multiplicatives: $A\otimes B=A\parr B=A\lolipop B=A\times B$

The category is cartesian, with catesian product $\bigwith A_i=\{(i,x)|i\in I, x\in s_i\}$, projections $\pi_i=\{((i,a),a)|a\in A_i\}$ and product of morphisms $\bigwith_i f_i=\{(b,(i,a))|(b,a)\in f_i\}$. The terminal object is $\top =\emptyset$.\\
This give us the interpretation of additives: $\bigoplus_{i\in I}A_i=\bigwith_{i\in I}A_i=\{(i,a)|i\in I, a\in A_i\}$

We can add a comonade $(!,d,p)$ where the functor is define by $!A=\mathcal{M}_f(A)$, $!f=\{([a_1,...,a_k],[b_1,...,b_k)|\forall i\le k,(a_i,b_i)\in f\}$, the deriliction by $d_A=\{([a],a)|a\in A\}$ and the digging by $p_A=\{(m_1\+\cdots\+m_k, [m_1,...,m_k])|m_1,...,m_k\in \mathcal{M}_F(a)\}$.\\
This give us the interpretation of exponentials: $!P=?P=\mathbb{M}_f(P)$

This is a Seely category and a model of linear logic since the isomorphismes $1\simeq []$ is trivial and $!A\otimes !B\simeq !(A\with B)$ is defined by $([a_1,...,a_l],[b_1,...,b_r])\simeq [(1,a_1),...,(1,a_l),(2,b_1),...,(2,b_r)]$. 

But it can even be seen as a categorical model of differential linear logic.
By defining the co-dereliction natural transforamtion $\bar d_A=\{(a,[a])|a\in A\}\in A\rta !A$, we are fixing the contraction $c_A=\{(l\+r,(l,r))|l,r\in!A\}$ the co-contraction $\bar c_A=\{((l,r),l\+r)|l,r\in !A\}$, the weakening $w_A=\{([],*)\}$ and the co-weakening $\bar w_A=\{(*,[])\}$. So that we can define the derivative $\partial_X=(id_{!X}\otimes d_X)\circ c_X: !X\rta !X\otimes X$ and the co-derivative $\bar\partial_X=\bar c_X\circ(id_{!X}\otimes\bar d_D):!X\otimes X\rta !X$. This derivatives are Taylor, {\em i.e.} if two morphisms $f_1,f_2: !A\rta B$ are such that $f_1\circ \bar\partial_A=f_2\circ\bar\partial_A$ then $f_1+(f_2\circ\bar w_X\circ w_X)=(f_1\circ\bar w_X\circ w_X)+f_2$. Finally the exponential acept anti-derivatives, since it is bi-comutative and $J_A=Id_A+\bar\partial_A\partial_A=Id_A$ is an isomorphism. For more detail about models of DLL see \ref{Ehr11}.

As for every categorical model of linear logic, the exponential is a comonade and induced a coKleisly $MRel=Rel_!$ whose objects are the set and whose morphisms from $P$ to $Q$ are the relations between $!P$ and $Q$. The identities are the relations $dig_P=\{(\{x\},x)|x\in P\}$ and the composition $f\circ g=\{(X,z)|\exists (Y,z)\in f, \forall y\in Y, (X,y)\in g\}$.

\subsection{Algebraic model}

In order to have an algebraic model of \dlcalcul we only need a reflexive object, {\em i.e.} a triplet $(M,app,abs)$ where $M$ is an object of MRel, $app:(M\rta (!M^\bot\lolipop M)=(\mathcal{M}_f(M)\times M))$ and $aps:((\mathcal{M}_f(M)\times M)\rta M)$ such that $app\circ abs=Id$. Such an object can a priory found by taking the lower fix point of $M\mapsto M\Rta M=?M\!\lolipop M=\mathcal{M}_f(M)\times M$. But this will just leads to the trivial empty model. We will then resolve the more complicated fix point $M\mapsto (?M)^{\mathbb{\with N}}\!\!\!=\mathbb{N}\times \mathcal{M}_f(M)$ where the exponent represent an infinit tensor product. The lower fix point will be called $\Minf$.

An other way to see the fixpoint is to say that $M$ have to be equal to the set of quazi everywhere empty lists of finite substets of itself. Its element are the recursively defined as being either $*$, the list of empty elements, or $a\cons\alpha$ with $a\in\mathcal{M}_f(\Minf)$ and $\alpha\in\Minf$. The coresponding $app$ and $abs$ arise imediatly from the functoriality: $app=\{(a\cons\alpha,(a,\alpha))|(a,\alpha)\in \Minf\}$ and $abs=\{((a,\alpha),a\cons\alpha)|a,\alpha\in \Minf\}$

In order to be understandable, we are presenting the interpretation of terms via a type system with types living in $\Minf$. The usual presentation of the interpretation can be recoverd from the type system:\\
$\llb M\rrb^{\bar x}=\{(\bar a,\alpha)|\bar x:\bar a\vdash M:\alpha\}$\\
$\llb Q\rrb^{\bar x}=\{\bar a|\bar x:\bar a\vdash Q\}$\\
The type system is the following:

\hskip -2.5cm
\begin{minipage}[b]{18cm}
\begin{center}
\AxiomC{$\Gamma\vdash \Delta$}
\UnaryInfC{$x:[],\Gamma\vdash \Delta$}
\DisplayProof\hskip 20pt
\AxiomC{}
\UnaryInfC{$x:[\alpha]\vdash x:\alpha$}
\DisplayProof\hskip 20pt
\AxiomC{$\Gamma\vdash M:\alpha$}
\UnaryInfC{$\Gamma\vdash M\+N:\alpha$}
\DisplayProof\hskip 20pt
\AxiomC{$\Gamma\vdash N:\alpha$}
\UnaryInfC{$\Gamma\vdash M\+N:\alpha$}
\DisplayProof\hskip 20pt
\AxiomC{$\Gamma,x:v\vdash M:\alpha$}
\UnaryInfC{$\Gamma\vdash \lambda x.M:v\cons\alpha$}
\DisplayProof\\
\AxiomC{$\Gamma\vdash M:w\cons\alpha$}
\AxiomC{$\Gamma'\vdash B:w$}
\BinaryInfC{$\Gamma+\Gamma'\vdash  M\ B:\alpha$}
\DisplayProof\hskip 20pt
\AxiomC{$\bigwedge_{j\le n}\ \Gamma_j\vdash L_j:\beta_j$}
\AxiomC{$\bigwedge_{i\ge n}\ \Gamma_i\vdash L:\beta_i$}
\BinaryInfC{$(\Sigma^{n+m}_{r=1}\Gamma_r)\vdash [L_1,...,L_n;L^!]:[\beta_1,...,\beta_{n+m}]$}
\DisplayProof\\
\AxiomC{$\Gamma\vdash Q$}
\UnaryInfC{$\Gamma\vdash  \bar \tau (Q):*$}
\DisplayProof\hskip 50pt
\AxiomC{$\Gamma\vdash M:*$}
\UnaryInfC{$\Gamma\vdash  \tau [M]$}
\DisplayProof\hskip 50pt
\AxiomC{$\Gamma\vdash Q$}
\AxiomC{$\Gamma'\vdash R$}
\BinaryInfC{$\Gamma+\Gamma'\vdash  Q|R$}
\DisplayProof\hskip 50pt
\AxiomC{}
\UnaryInfC{$\vdash \epsilon$}
\DisplayProof \\
\end{center}
\end{minipage}

\end{document}